\documentclass[aps,final,twocolumn,floatfix,showpacs]{revtex4-1}

\usepackage{amsmath,amsfonts,amsthm}

\usepackage{graphicx,color,braket}

\usepackage{graphicx}
\usepackage{dcolumn}
\usepackage{bm}
\usepackage{amssymb}
\usepackage{natbib}
\usepackage{amsmath}
\usepackage{color}
\usepackage{datetime}
\usepackage{footnote}
\usepackage{bbold}
\usepackage{hyperref}
\usepackage[normalem]{ulem}
\usepackage{dcolumn}

\newtheorem*{lemma}{Lemma}
\newtheorem*{theorem}{Theorem}

\begin{document}

\title{Can different quantum state vectors correspond to the same physical state? An experimental test}

\author{Daniel Nigg$^{1}$}
\altaffiliation{* These authors contributed equally to this work.}
\author{Thomas Monz$^{1}$}
\altaffiliation{* These authors contributed equally to this work.}
\author{Philipp Schindler$^{1}$}
\author{Esteban A. Martinez$^{1}$}
\author{Michael Chwalla$^{2}$}
\author{Markus Hennrich$^{1}$}
\author{Rainer Blatt$^{1,2}$}
\author{Matthew F. Pusey$^{3}$}
\author{Terry Rudolph$^{3}$}
\author{Jonathan Barrett$^{4}$} 

\affiliation{$^{1}$Institut f\"ur Experimentalphysik, Universit\"at Innsbruck, Technikerstr. 25, A-6020 Innsbruck, Austria\\
$^{2}$Institut f\"ur Quantenoptik und Quanteninformation, \"Osterreichische Akademie der Wissenschaften, Technikerstr.~21A, A-6020 Innsbruck, Austria\\
$^{3}$Department of Physics, Imperial College London, Prince Consort Road, London SW7 2AZ, United Kingdom
$^{4}$Department of Computer Science, University of Oxford, Parks Road, Oxford OX1 3QD United Kingdom}

\maketitle

\textbf{A century on from the development of quantum theory, the interpretation of a quantum state is still discussed. If a physicist claims to have produced a system with a particular wave function, does this represent directly a physical wave of some kind, or is the wave function merely a summary of knowledge, or information, about the system? A recent no-go theorem shows that models in which the wave function is not physical, but corresponds only to an experimenter's information about a hypothetical real state of the system, must make different predictions from quantum theory when a certain test is carried out. Here we report on an experimental implementation using trapped ions. Within experimental error, the results confirm quantum theory. We analyse which kinds of theories are ruled out.}

In the development of quantum theory, progress has been made when apparently vague questions are made precise, and a no-go theorem proved. The best known no-go theorem is perhaps Bell's theorem \cite{bell}, which shows that (in Bell's terminology) \emph{locally causal} theories cannot reproduce the quantum predictions for measurements performed on separated entangled systems. Many experiments have confirmed the quantum predictions in this case. Strictly speaking there remain loopholes, which no experiment to date has closed simultaneously. It is nonetheless safe to say that all but highly conspiratorial locally causal models have been ruled out.

Similarly, Ref\cite{PBR} (following \cite{nic} and \cite{hpriv}) proves a no-go theorem that shines light on the controversies above. The question addressed is a rather simple one: can two distinct quantum state vectors sometimes describe the same physical reality? Under the assumption that independently prepared systems have independent physical states, it is shown that the answer is no -- any such model must make different predictions from quantum theory when a certain carefully designed test is performed. 

This work reports an experimental implementation of the test from Ref.\cite{PBR} using trapped ions. In the following, we introduce the question in more detail and describe briefly the original proposal of Ref.~\cite{PBR}. A real experiment (with inevitable imperfections) will never be able to rule out the full class of models that the theorem states. Instead, as explained below, a well conducted experiment is able to place bounds on the degree to which similar but distinct quantum states can represent similar information about the system. We identify a natural subclass of models that our experiment is able to rule out. The experimental setup is then described, followed by the experimental results and their interpretation. 

\section{What is information and what is physical?}

Why isn't the conclusion of Ref.~\cite{PBR} -- that distinct quantum states correspond to distinct states of reality -- obvious? How could it be that they sometimes correspond to the same state of reality, given that via the Born rule, different quantum states lead to different experimental predictions?

Imagine a die shaken randomly in a container that has internal components which prevent the die being removed unless the value showing upperside is even. The filtering action of the container ensures the probability distribution one would assign to the upperside value of a die (out of [1,2,3,4,5,6]) taken from it is $p_{\mathrm{even}}=[0,1/3,0,1/3,0,1/3]$. Imagine a similar container exists that contrives to ensure the value shown is prime, so the probability distribution assigned is $p_{\mathrm{prime}}=[0,1/3,1/3,0,1/3,0]$. Given a die and asked to determine which container was used to prepare it, one is not certain to succeed, because one particular configuration of the die -- the number 2 -- is consistent with either preparation procedure.

Now consider a quantum system that can be prepared in two different ways corresponding to non-orthogonal quantum states. Here again, given the quantum system and asked to determine which preparation procedure was used, one cannot always succeed. As is well known, it is not possible to discriminate between non-orthogonal quantum states with certainty.

In the example of the die, the reason why one cannot guess the preparation with certainty is because the same state of reality (the number 2) is assigned non-zero probability by each of the two probability distributions, $p_{\mathrm{even}}$ and $p_{\mathrm{prime}}$. The impossibility of distinguishing non-orthogonal quantum states would receive a very natural explanation if something similar were true -- if the two different quantum states did, in fact, sometimes describe the same underlying physical state of the system.

Let us hypothesise, therefore, that this is the case and consider the following schematic account of how a physical theory might describe a quantum system. It is similar to the formalism used in derivations of Bell's theorem, with the difference that we will not assume locality. First, the quantum system has an objective physical state of some kind, denoted $\lambda$, where objective means that $\lambda$ is independent of the experimenter and of other physical systems. For our purposes, this assumption only needs to hold for systems prepared in isolation and assigned a pure quantum state. If a measurement is performed on the system, the probabilities for different outcomes are determined by $\lambda$. Given an ensemble of systems, each prepared in such a way that quantum theory assigns a quantum state $\ket{\psi}$, it is not necessarily the case that $\lambda$ is the same for each member of the ensemble, hence we assume that a quantum state $\ket{\psi}$ corresponds to a probability distribution $\mu_\psi(\lambda)$.

The possibility raised above is that the distributions $\mu_0(\lambda)$ and $\mu_1(\lambda)$ overlap for distinct quantum states $\ket{\psi_0}$ and $\ket{\psi_1}$. It is, however, a non-trivial question whether models of this form actually exist which reproduce the predictions of quantum theory. The question has been raised by Harrigan and Spekkens \cite{nic} and by Hardy \cite{hpriv}.  Harrigan and Spekkens refer to any model in which the distributions $\mu_0(\lambda)$ and $\mu_1(\lambda)$ overlap for some non-identical $\ket{\psi_0}$ and $\ket{\psi_1}$ as $\psi$-\emph{epistemic}. A model in which $\mu_0(\lambda)$ and $\mu_1(\lambda)$ are disjoint for any non-identical $\ket{\psi_0}$ and $\ket{\psi_1}$ is called $\psi$-\emph{ontic}. For a single qubit, a $\psi$-\emph{epistemic} model was provided by Kochen and Specker~\cite{ksmodel}. For larger--dimensional Hilbert spaces, a model wherein some quantum states are represented by overlapping distributions is described in Ref.~\cite{ljbr}, which was recently extended to a model where all quantum states overlap \cite{scott}.

Quantum theory, however, allows for more complicated protocols than separate measurements on single systems. In Ref.~\cite{PBR}, such a protocol was identified, for which the quantum predictions cannot be reproduced using a $\psi$-epistemic model. The protocol begins with the independent preparation of $n$ systems, each according to one of two pure quantum states $\ket{\psi_0}$and $\ket{\psi_1}$. The joint quantum state is a direct product $\ket{\psi_{x_1}}\otimes \ket{\psi_{x_2}}\otimes\cdots\otimes \ket{\psi_{x_n}}$, where $x_i\in \{0,1\}$. Assume that these systems have separate underlying physical states $\lambda_1,\ldots ,\lambda_n$. Importantly, the model also assumes that when systems are prepared independently, the physical states are independent from one another, meaning that the joint distribution is given by~\footnote{The model of joint probability distributions can be slightly weakened~\cite{hall,fine}, or replaced by an entirely different assumption~\cite{crthm,hardythm}}
\begin{equation}\label{independence}
\mu_{x_1}(\lambda_1)\times \mu_{x_2}(\lambda_2)\times\cdots\times \mu_{x_n}(\lambda_n).
\end{equation}
The systems are then brought together, and an entangled joint measurement is made such that according to quantum theory, the first outcome is zero if the quantum state is $\ket{\psi_0}\otimes\ket{\psi_0}\otimes\dotsb\otimes\ket{\psi_0}$, the second outcome is zero if the quantum state is $\ket{\psi_0}\otimes\ket{\psi_0}\otimes\dotsb\otimes\ket{\psi_1}$, and so on. It is shown in Ref.~\cite{PBR} that for any distinct $\ket{\psi_0}$, $\ket{\psi_1}$, such a measurement exists for large enough $n$. These predictions are incompatible with the quantum state vectors $\ket{\psi_0}$ and $\ket{\psi_1}$ being represented by overlapping distributions $\mu_0(\lambda)$ and $\mu_1(\lambda)$, which would always result in non-zero probabilities.

This argument relies on certain probabilities being exactly zero, which will never be experimentally reproduced. What can be concluded if these probabilities are merely close to zero? 
Consider again a classical system, such as the die above, which is equally likely to have been prepared according to one of two probability distributions $p=(p_1,\ldots,p_n)$ and $q=(q_1,\ldots ,q_n)$. The maximal probability of guessing correctly which preparation was used is $(1+D(q,p))/2$, where the \emph{trace distance} $D$ is defined by
\begin{equation}
D(q,p)=\frac{1}{2}\sum_i |q_i-p_i|.
\end{equation}
More generally, the trace distance between two probability distributions is defined by a similar expression, with the sum replaced by an integral if necessary. The trace distance is $1$ if the distributions are completely disjoint (and thus perfectly distinguishable), and it is $0$ if they are identical. The \emph{quantum trace distance} for pure states is
\begin{equation}
D_Q(\ket{\psi_0}, \ket{\psi_1}) = \sqrt{1 - \left\lvert{\braket{\psi_0 | \psi_1}}\right\rvert^2}.
\end{equation}
Given quantum states $\ket{\psi_0}$ and $\ket{\psi_1}$, which are a priori equally likely, if an optimal measurement for discriminating the two is performed, the probability of guessing correctly is $(1+D_Q(\ket{\psi_0},\ket{\psi_1}))/2$. This implies that any $\psi$-epistemic model capable of reproducing the quantum predictions for the state discriminating measurement must satisfy
\begin{equation}\label{maxepiscondition}
D(\mu_0,\mu_1) \geq D_Q(\ket{\psi_0}, \ket{\psi_1}) .
\end{equation}

The protocol described in Ref.\cite{PBR} goes further. It shows that for any pair of distinct quantum states $\ket{\psi_0}$ and $\ket{\psi_1}$, if a suitable experiment is performed, then a lower bound can be placed on $D(\mu_0,\mu_1)$. The better the data match quantum predictions, the closer this bound is to $1$. As might be expected, the experiment is harder to perform for states $\ket{\psi_0}$ and $\ket{\psi_1}$ that are close together.

The aim of this work is to rule out experimentally the case in which 
\begin{equation}\label{equaltracedistances}
D(\mu_0,\mu_1) = D_Q(\ket{\psi_0},\ket{\psi_1}).
\end{equation}
This corresponds to a natural class of models, since if Eq.(\ref{equaltracedistances}) is satisfied, the probability of making an error when attempting to discriminate non-orthogonal quantum states does not need to be explained by any kind of uniquely quantum effect, but is in fact \emph{entirely due} to the ordinary classical difficulty of distinguishing the corresponding distributions $\mu_0$ and $\mu_1$. 
An example of a model that satisfies this requirement is the qubit model of Kochen and Specker mentioned above \cite{lpriv}. (The idea that the impossibility of reliably distinguishing non-orthogonal quantum states is entirely due to overlapping $\mu_0$ and $\mu_1$ has also been considered in Refs.\cite{morris, maroney, leifermaroney}, but with alternative definitions of the distances between classical probability distributions.)   

We report a careful implementation of the protocol of Ref.~\cite{PBR} such that probabilities inferred from the experimental data are close to zero when the quantum prediction is zero. As mentioned above, it is then possible to derive a lower bound on the trace distance $D(\mu_0,\mu_1)$. In the appendix, we derive an alternative lower bound to that of Ref.~\cite{PBR}, specially adapted to the present experimental context. Using this, our experimental results rule out models satisfying Eq.(\ref{equaltracedistances}).

\section{Experimental Setup}

The experiments described here were realized in a system consisting of $^{40}$Ca$^{+}$ ions
which are confined to a string in a linear Paul trap~\cite{exp}.
Each ion represents a qubit which is encoded in the electronic levels $S_{1/2}(m
= -1/2)=\ket{1}$ and $D_{5/2}(m = -1/2)=\ket{0}$. Each experimental cycle consists of an
initialization of the ions in their internal electronic and motional
ground states followed by a coherent manipulation of the qubits and
finally a detection of the quantum state. State initialization is
realized by optical pumping into the $S_{1/2}(m = -1/2)$ state
after cooling the axial centre-of-mass mode to the motional
ground state. The manipulation of the qubits is implemented by
coherently exciting the $S_{1/2} \leftrightarrow D_{5/2}$ quadrupole
transition with laser pulses. Finally, the population of the qubit
states is measured by exciting the $S_{1/2} \leftrightarrow
P_{1/2}$ transition and detecting the fluorescence, using
electron shelving~\cite{Dehmelt}. Our setup is capable of realizing
collective qubit rotations
\begin{equation}
  U(\theta,\phi) =
\exp\left(-i\frac{\theta}{2}\sum_{i}\left[\sin(\phi)\sigma_{y}^{(i)}-\cos(\phi)\sigma_{x}^{(i)}\right]\right)
\end{equation}
via a laser beam addressing the entire register as well as
M{\o}lmer-S{\o}renson entangling gates \cite{ms,ms_exp}
\begin{equation}
  MS(\theta,\phi) =
\exp\left(-i\frac{\theta}{4}\biggl[\sum_{i}(\sin(\phi)\sigma_{y}^{(i)}-\cos(\phi)\sigma_{x}^{(i)})\biggr]^{2}\right).
\end{equation}
Additionally we are able to perform single-qubit rotations on
the $i$-th ion of the form $U_{z}^{(i)}(\theta) =
\exp\left(-i\frac{\theta}{2}\sigma_{z}^{(i)}\right)$ using an off-resonant laser
beam which addresses individual ions.

The first step of the protocol is the preparation of one of the
input states $\{\ket{\phi_{0}}\otimes\ket{\phi_{0}}, \ket{\phi_{0}}\otimes\ket{\phi_{1}},
\ket{\phi_{1}}\otimes\ket{\phi_{0}},\ket{\phi_{1}}\otimes\ket{\phi_{1}}\}$ with
$\ket{\phi_{0}} = \cos(\frac{\pi}{8})\ket{1} + \sin(\frac{\pi}{8})\ket{0}$ and $\ket{\phi_{1}} = \cos(\frac{\pi}{8})\ket{1} - \sin(\frac{\pi}{8})\ket{0}$.
The input states are generated by a global rotation $U(\frac{\pi}{4},-\frac{\pi}{2})$, which maps $\ket{1}\otimes\ket{1}$ to $\ket{\phi_0}\otimes\ket{\phi_0}$, followed by
single-qubit rotations $U_{z}^{(1,2)}(\pi)$ on the first, second or both ions, which maps $\ket{\phi_0}$ to $\ket{\phi_1}$. These sequences are shown in Fig.~\ref{sequ}(b).

\begin{figure}
\includegraphics[scale=0.8]{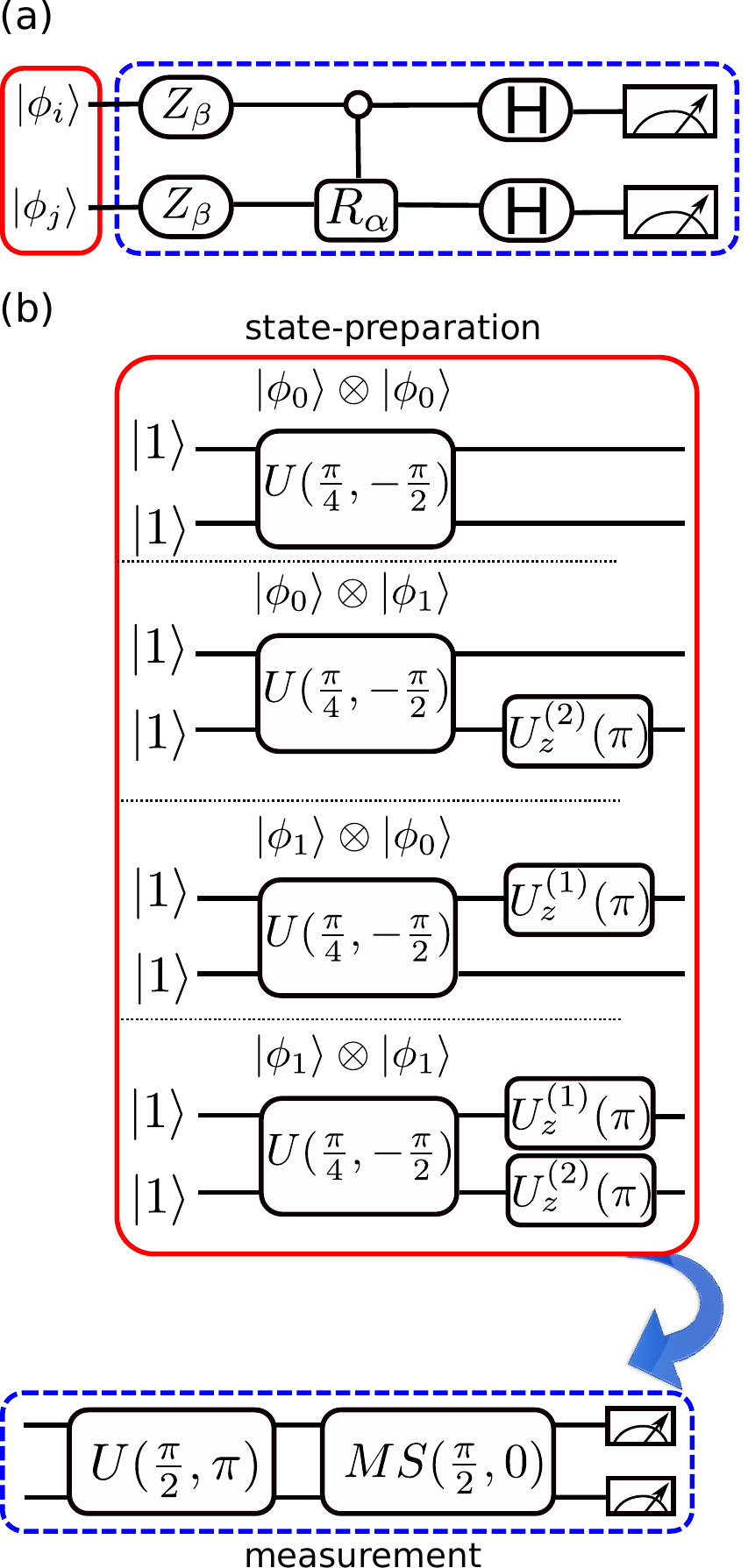}
\caption{Experimental implementation of (a) the measurement protocol consisting of the input state
preparation followed by a joint measurement of the whole system and the (b) corresponding
pulse-sequence.}
\label{sequ}
\end{figure}

The second part of the protocol is a joint measurement on the system with the property that each
measurement outcome should have probability zero for one of the four input states. Such a measurement can, in general, be realized
by rotations $Z_{\beta} = \ket{0}\bra{0}+\exp(i\beta)\ket{1}\bra{1}$ followed by a
conditional phase gate $R_{\alpha}$, with $R_{\alpha} \ket{11} = \exp(i\alpha) \ket{11}$, Hadamard gates H and finally
a measurement in the computational basis, shown in Fig.~\ref{sequ}(a).

Our pulse sequence for the measurement procedure is shown in Fig.~\ref{sequ}(b). For the input states used here, $\alpha = \pi$ and $\beta = 0$, which reduces $Z_{\beta}$ to the identity.
Up to local phases (which do not matter because the final step is a measurement in the computational basis) the remaining two operations are equivalent to a global rotation $U(\frac{\pi}{2},\pi)$ followed by a maximally entangling gate $MS(\frac{\pi}{2},0)$.

In an ideal experiment, the single ion pulses used to select different input states would have no effect on the other ion. This requirement is in fact only fulfilled to a certain
degree of accuracy due to residual light on the neighbouring ion. If, for example,
the phase shift operation $U^{(2)}_{z}(\pi)$ is applied on the second qubit, a residual phase shift
$U^{(1)}_{z}(\kappa\pi)$ occurs on the first ion, with $\kappa$ on the order of 1\%. Hence, instead of the ideal input state $\ket{\phi_{0}}\otimes\ket{\phi_{1}}$, we have $\ket{\phi_{0}'}\otimes\ket{\phi_{1}}$, where $\ket{\phi_{0}'} = e^{i\kappa\pi}\cos\left( \frac\pi8 \right)\ket{1} + \sin\left(\frac\pi8\right)\ket{0}$. Due to the residual light in both directions, $\ket{\phi_1}\otimes\ket{\phi_1}$ becomes $\ket{\phi_1'}\otimes\ket{\phi_1'}$ where $\ket{\phi_{1}'} = e^{i\kappa\pi}\cos\left( \frac\pi8 \right)\ket{1} - \sin\left(\frac\pi8\right)\ket{0}$.

Overall, then, in place of the original protocol's preparations $\ket{\phi_0}\otimes\ket{\phi_0}$, $\ket{\phi_0}\otimes\ket{\phi_1}$, $\ket{\phi_1}\otimes\ket{\phi_0}$, $\ket{\phi_1}\otimes\ket{\phi_1}$, we instead have $\ket{\phi_0}\otimes\ket{\phi_0}$, $\ket{\phi_0'}\otimes\ket{\phi_1}$, $\ket{\phi_1}\otimes\ket{\phi_0'}$, $\ket{\phi_1'}\otimes\ket{\phi_1'}$. Since the
laser field can be defined as a product of coherent states across the ions, it is reasonable to model these as independent preparations:
\begin{equation} \mu_{00}(\lambda_1, \lambda_2) = \mu_0(\lambda_1)\mu_0(\lambda_2), \end{equation}
\begin{equation} \mu_{0'1}(\lambda_1, \lambda_2) = \mu_{0'}(\lambda_1)\mu_1(\lambda_2), \end{equation}
\begin{equation} \mu_{10'}(\lambda_1, \lambda_2) = \mu_1(\lambda_1)\mu_{0'}(\lambda_2), \end{equation}
\begin{equation} \mu_{1'1'}(\lambda_1, \lambda_2) = \mu_{1'}(\lambda_1)\mu_{1'}(\lambda_2). \end{equation}

It is shown in the supplementary information that if this modified protocol is implemented, and the outcomes that never occur in the ideal protocol happen with probability $\epsilon$, then the relevant classical trace distances satisfy

\begin{equation}
  D(\mu_0, \mu_1) + D(\mu_0, \mu_{0'}) + D(\mu_1, \mu_{1'}) \geq 1 - 2\sqrt{\epsilon}.
  \label{Deq}
\end{equation}

Supposing that all three classical trace distances were equal to the corresponding quantum trace distances, we would obtain $\epsilon \geq 1.83\%$, and so an experiment with a smaller value rules out the $\psi$-epistemic hypothesis. For simplicity, we are treating the ions identically, whereas in fact the actions of the laser on each ion will be slightly different. In the supplementary information it is shown that this consideration does not significantly alter the lower bound on $\epsilon$.

Ideally, each outcome would have probability zero for some input state. In fact these nonzero outcomes occur with small probabilities $\epsilon_{k}$. For each of the input states the measurement was repeated up to 10000 times to gain significant statistics. Fig.~\ref{result} shows the measured probabilities.
We find $\{\epsilon_1,\epsilon_2, \epsilon_3, \epsilon_4\} =\{1.88\%, 1.04\%, 0.6\%, 1.04\%\}$,
which yields a mean value of $\epsilon = (1.14\pm0.15)\%$. This violates the bound of Eq.~\eqref{Deq} given by our trace distance hypothesis by over 4.5$\sigma$.
Therefore the probability of this data to be consistent with a $\psi$-\emph{epistemic} model is less than $10^{-5}$.

\begin{figure*}
\includegraphics[scale=0.9]{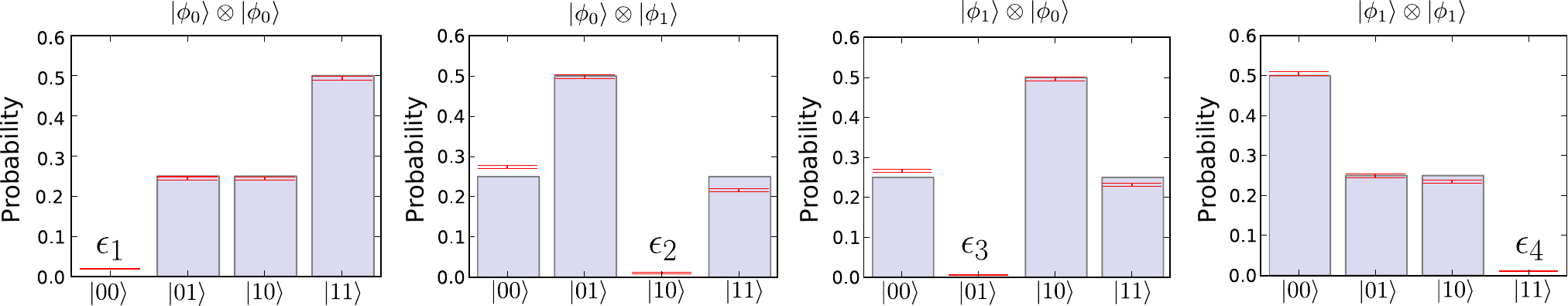}
\caption{Experimental measured probabilities for each input state $\ket{\phi_{i}}\otimes\ket{\phi_{j}} ((i,j)\in\{0,1\})$ (red lines) versus quantum predictions (grey bars). Ideally, for each input state one probability $\epsilon_{k}$ is expected to be zero. The experimental measured probabilities $\{\epsilon_1,\epsilon_2, \epsilon_3, \epsilon_4\} =\{1.88\pm 0.5 \%, 1.04\pm 0.10\%, 0.6\pm 0.08\%, 1.04\pm 0.50\%\}$ are close to zero and therefore we rule out the $\psi$-epistemic model by over 4$\sigma$. The uncertainty of each measurement outcome is of statistical nature and based on projection noise.}
\label{result}
\end{figure*}

\section{Discussion}

The result of Ref.~\cite{PBR} can be considered as a no-go theorem for interpretations of quantum theory, analogous to Bell's theorem. Each theorem states that a certain class of theories must make different predictions from quantum theory -- \emph{locally causal} theories in the case of Bell's theorem and $\psi$-\emph{epistemic} theories in the case of Ref.~\cite{PBR}.
We have suggested that a natural threshold is defined by quantum state discrimination. Given the assumptions of the analysis, our experimental data rules out models that satisfy Eq.(\ref{equaltracedistances}), i.e., those in which errors in quantum state discrimination are explained entirely in terms of overlapping classical probability distributions.

What assumptions are needed? First, both Bell's theorem and Ref.~\cite{PBR} assume that a system has an objective physical state, to which it makes sense to attach the label $\lambda$. We will not say more about this -- it is the possibility of recovering experimental data under such an assumption that is being explored. Second, the theorem of Ref.~\cite{PBR} needs to assume \emph{preparation independence}, that is, Eq.(\ref{independence}). An ideal implementation of the protocol of Ref.~\cite{PBR} would involve preparations of quantum systems that are independent by any reasonable scientific judgement (say, one on Earth and one on Mars, spacelike separated, using apparatuses manufactured in separate factories, each unaware of the other...). In this case, we would not regard the assumption of preparation independence as strong: it is a basic tenet of physical science that one can perform independent experiments, and when this is done, the relevant systems are uncorrelated. An assumption needed by Bell's theorem, that experimenters can freely choose measurements, has a similar character. In each case, models can be constructed that violate the assumption and evade the theorem, but these are highly contrived~\cite{indep,indep2,ljbr}.

Real experiments, however, deviate from the ideal in various ways. 
In our experiment, the quantum systems are ions in close proximity. The judgement that these are independent systems, and that Eq.(\ref{independence}) is reasonable, derives from the fact that quantum theory itself assigns a product pure state (i.e., the judgement does \emph{not} come from more general background considerations, as in the ideal case). Furthermore, the decision whether to prepare one ion in the state $\ket{\psi_0}$ or $\ket{\psi_1}$ has a measurable effect on the quantum state of the other ion, hence these cannot be regarded as truly independent preparations and in that sense corresponds to a potential state preparation loophole. This is reflected in the inequality~\eqref{Deq}, which is designed to allow for the fact that when one ion is rotated from the state $\ket{\psi_0}$ to the state $\ket{\psi_1}$, a small crosstalk effect changes the state of the other ion from $\ket{\psi_0}$ to $\ket{\psi_1'}$. The inequality \eqref{Deq} allows the corresponding distributions $\mu_0(\lambda)$ and $\mu_1(\lambda)$ to be different.

However, the derivation of inequality \eqref{Deq} does still assume that a preparation of a quantum state which is (close to) a direct product, e.g., $\ket{\psi_0}\otimes \ket{\psi_1}$, corresponds to a distribution which is also (close to) a product, e.g., $\mu_0(\lambda_1)\mu_1(\lambda_2)$. Given the crosstalk effect, a model which violates this assumption is not as farfetched as such a model would be, were the preparations completely independent at the quantum level. Future experiments may be able to achieve this.

Finally, and perhaps most significantly, the quantum states prepared will not in fact have been completely pure. This will not affect the bound on classical trace distances that we obtained, but it is not clear exactly what a natural hypothesis for these classical trace distances would be. One certainly cannot ask that the classical trace distance is equal to the quantum trace distance in the case of mixed states, because sometimes even identical mixed quantum states cannot be represented by the same classical probability distribution~\cite{spekkens2}. This opens
the possibility for a ``mixed states'' preparations loophole.

The fundamental question of the interpretation of the wave function of a quantum system was partially
resolved. If systems have real states, regardless of an experimenter
or measurements performed, then a natural question is whether the quantum state is epistemic, i.e. corresponding
merely to knowledge of these underlying real states. In the presented manuscript we tested for this specific
possibility and ruled out the most natural class of such models to a high degree of confidence.

\begin{acknowledgments}
We thank Tracy Northup and Mike Brownnutt for critically reading our manuscript.
We gratefully acknowledge support by the Austrian Science Fund (FWF), through the SFB FoQus (FWF Project No. F4006-N16),
by the European Commission (AQUTE), by IARPA as well as the Institut f\"{u}r Quantenoptik und Quanteninformation GmbH. We also acknowledge support
from the UK Engineering and Physical Sciences Research Council, and the CHIST-ERA DIQIP project. Research at Perimeter Institute is supported by
the Government of Canada through Industry Canada and by the Province of Ontario through the Ministry of Research and Innovation.
\end{acknowledgments}

\pagebreak

\section{Appendix}

\section{Supplementary information: Main result}
  \begin{theorem}
    Let $\ket{\psi_0}, \ket{\psi_0'}, \ket{\psi_1}, \ket{\psi_1'}$ be qubit states. Suppose 2 qubits are independently prepared in one of the four states
  \begin{equation}
  \ket{\psi_0}\otimes\ket{\psi_0}, \ket{\psi_0'}\otimes\ket{\psi_1}, \ket{\psi_1}\otimes\ket{\psi_0'}, \ket{\psi_1'}\otimes\ket{\psi_1'}
  \end{equation}
  and a measurement is performed where outcome $k$ has probability $\epsilon_k$ on the $k$-th quantum state. Let $\epsilon = \frac14 \sum\limits_{i=1}^4 \epsilon_k$. A model that reproduces these results must satisfy
  \begin{equation}
  D(\mu_0, \mu_1) + D(\mu_0, \mu_{0'}) + D(\mu_1, \mu_{1'}) \geq 1 - 2\sqrt{\epsilon}
  \end{equation}
  \end{theorem}

To prove this we define the $k$-overlap of probability distributions $\mu_1, \dotsc, \mu_k$, as in \cite{PBR}, by
\begin{equation}
\omega(\mu_1,\ldots,\mu_k) = \int_\Lambda \min_i \mu_i(\lambda) d\lambda.
\end{equation}

We will need the following link between 4-overlaps and 2-overlaps:
\begin{lemma}\label{4o2o}
\begin{multline}
  \omega(\mu_A, \mu_B, \mu_C, \mu_D) \geq \\\omega(\mu_A, \mu_B) + \omega(\mu_B, \mu_C) + \omega(\mu_C, \mu_D) - 2.
\end{multline}
\end{lemma}
\begin{proof}
  First notice that for real numbers $a,b,c$
  \begin{equation}
    \min(a,b,c) = \min(a,b) + \min(b,c) - f(a,b,c),
  \end{equation}
  where
  \begin{equation}
    f(a,b,c) =
    \begin{cases}
      \min(b,c) & a \leq b,c \\
      b & b \leq a,c\\
      \min(a,b) & c \leq a,b
    \end{cases}.
  \end{equation}
  Since $f(a,b,c) \leq b$ we have
  \begin{equation}
    \min(a,b,c) \geq \min(a,b) + \min(b,c) - b.
  \end{equation}
  By two applications of the above we obtain that for real numbers $a,b,c,d$
  \begin{multline}
    \min(a,b,c,d) = \min\left( a,b,\min(c,d) \right) \geq \\ \min(a,b) + \min\left( b, \min(c,d) \right) - b =\\ \min(a,b) + \min(b,c,d) - b \\\geq \min(a,b) + \min(b,c) + \min(c,d) - b - c.
  \end{multline}
  Let $a = \mu_A(\lambda), b = \mu_B(\lambda), c = \mu_C(\lambda), d=\mu_D(\lambda)$, integrate each side with respect to $\lambda$, and recall that $\mu_B$ and $\mu_C$ are normalised to obtain the result.
\end{proof}

  \begin{proof}[Proof of Theorem]
  Since the systems are prepared independently, the resulting physical states are distributed according to
  \begin{align}
  \mu_A(\lambda_1, \lambda_2) &= \mu_0 (\lambda_1) \times \mu_0 (\lambda_2),\\
  \mu_B(\lambda_1, \lambda_2) &= \mu_{0'} (\lambda_1) \times \mu_1 (\lambda_2),\\
  \mu_C(\lambda_1, \lambda_2) &= \mu_1 (\lambda_1) \times \mu_{0'} (\lambda_2),\\
  \mu_D(\lambda_1, \lambda_2) &= \mu_{1'} (\lambda_1) \times \mu_{1'} (\lambda_2).
\end{align}
Hence
\begin{multline}
  \min_{k \in \{A,B,C,D\}} \mu_k(\lambda_1, \lambda_2) \geq \\\min\{\mu_0(\lambda_1), \mu_{0'}(\lambda_1), \mu_1(\lambda_1), \mu_{1'}(\lambda_1)\} \times \\
\min\{\mu_0(\lambda_2), \mu_{0'}(\lambda_2), \mu_1(\lambda_2), \mu_{1'}(\lambda_2)\}.
\end{multline}
Integrating both sides
\begin{equation}
\omega(\{\mu_k\}) \geq \left( \omega(\mu_0, \mu_{0'},\mu_1, \mu_{1'}) \right)^2,
\label{multisystem2single}
\end{equation}
which square roots to
\begin{equation}
\sqrt{\omega(\{\mu_k\})} \geq \omega(\mu_0, \mu_{0'},\mu_1, \mu_{1'}).
\end{equation}
Applying the Lemma (notice that we can freely re-order the arguments of $\omega$) we obtain
\begin{equation}
\sqrt{\omega(\{\mu_k\})} \geq \omega(\mu_0, \mu_1) + \omega(\mu_0, \mu_{0'}) + \omega(\mu_1, \mu_{1'}) - 2.
\end{equation}
Recalling that the classical trace distance $D(\mu_0, \mu_1) = 1 - \omega(\mu_0, \mu_1)$ this becomes
\begin{equation}\label{tracesum}
\sqrt{\omega(\{\mu_k\})} \geq 1 - D(\mu_0, \mu_1) - D(\mu_0, \mu_{0'}) - D(\mu_1, \mu_{1'}).
\end{equation}
Meanwhile, from the observed probabilities we have that
\begin{equation}
\int_{\Lambda^n} \xi_{M,k}(\vec{\lambda}) \mu_{k}(\vec{\lambda}) \mathrm{d}\vec{\lambda} = \epsilon_k.
\end{equation}
Since $\min_i \mu_i(\vec{\lambda}) \leq \mu_k(\vec{\lambda})$, and both $\xi_{M,k}(\vec{\lambda})$ and $\mu_k(\vec{\lambda})$ are non-negative,
\begin{equation}
\int_{\Lambda^n} \xi_{M,k}(\vec{\lambda}) \min_{i} \mu_{i}(\vec{\lambda}) \mathrm{d}\vec{\lambda} \leq \epsilon_k.
\end{equation}
Finally, sum over $k$ and use the normalization $\sum_{k} \xi_{M,k}(\vec{\lambda}) = 1$ to obtain
\begin{equation}
\omega\left( \{ \mu_{k} \} \right) \leq 4 \epsilon.\label{fourepsilon}
\end{equation}
Square rooting each side gives
\begin{equation}\label{epsilonbound}
\sqrt{\omega\left( \{ \mu_{k} \} \right)} \leq 2 \sqrt{\epsilon}.
\end{equation}
Combining this with \eqref{tracesum} gives the desired result.
\end{proof}

\section{Non-identical qubits}
In the above it is assumed that the four quantum states $\ket{\psi_0}, \ket{\psi_0'}, \ket{\psi_1}$ and $\ket{\psi_1'}$ are the same for Alice and Bob, and the same underlying model of physical states $\lambda \in \Lambda$ is applied to both qubits. In reality the conditions of the qubits might not be identical, and so we should consider quantum states
  \begin{equation}
  \ket{\psi_0}\otimes\ket{\phi_0}, \ket{\psi_0'}\otimes\ket{\phi_1}, \ket{\psi_1}\otimes\ket{\phi_0'}, \ket{\psi_1'}\otimes\ket{\phi_1'},
  \end{equation}
  Corresponding to probability distributions over $\Lambda^a \times \Lambda^b$
  \begin{align}
  \mu_A(\lambda_1, \lambda_2) &= \mu^a_0 (\lambda_1) \times \mu^b_0 (\lambda_2),\\
  \mu_B(\lambda_1, \lambda_2) &= \mu^a_{0'} (\lambda_1) \times \mu^b_1 (\lambda_2),\\
  \mu_C(\lambda_1, \lambda_2) &= \mu^a_1 (\lambda_1) \times \mu^b_{0'} (\lambda_2),\\
  \mu_D(\lambda_1, \lambda_2) &= \mu^a_{1'} (\lambda_1) \times \mu^b_{1'} (\lambda_2).
\end{align}
Hence in place of Eq.~\eqref{multisystem2single} we obtain
\begin{equation}
\omega(\{\mu_k\}) \geq \omega(\mu^a_0, \mu^a_{0'},\mu^a_1, \mu^a_{1'}) \times \omega(\mu^b_0, \mu^b_{0'},\mu^b_1, \mu^b_{1'})
\end{equation}
Applying the Lemma and 2-overlaps to classical trace distances then gives
\begin{multline}
\omega(\{\mu_k\}) \geq \left(1 - D(\mu^a_0, \mu^a_1) - D(\mu^a_0, \mu^a_{0'}) - D(\mu^a_1, \mu^a_{1'})\right)\\\times \left(1 - D(\mu^b_0, \mu^b_1) - D(\mu^b_0, \mu^b_{0'}) - D(\mu^b_1, \mu^b_{1'})\right).
\end{multline}
Combining this with $\eqref{fourepsilon}$ gives
\begin{multline}
\left(1 - D(\mu^a_0, \mu^a_1) - D(\mu^a_0, \mu^a_{0'}) - D(\mu^a_1, \mu^a_{1'})\right)\\\times \left(1 - D(\mu^b_0, \mu^b_1) - D(\mu^b_0, \mu^b_{0'}) - D(\mu^b_1, \mu^b_{1'})\right) \leq 4\epsilon.
\end{multline}
If we suppose that, for both systems, the classical trace distance equals the quantum trace distance, and that (to good approximation) the quantum trace distances between the relevant states are as before, we obtain exactly the same bound on $\epsilon$.


\end{document}